\documentclass[table,xcdraw,dvipsnames]{lmcs}
\pdfoutput=1

\usepackage{lastpage}
\lmcsdoi{18}{1}{38}
\lmcsheading{}{\pageref{LastPage}}{}{}%
{Aug.~31,~2021}{Mar.~10,~2022}{}

\usepackage[utf8]{inputenc}


\usepackage{hyperref}
\usepackage[utf8]{inputenc}
\usepackage[T1]{fontenc}
\usepackage{url}           
\usepackage{booktabs}
\usepackage{amsfonts}
\usepackage{nicefrac}  
\usepackage{microtype}
\usepackage{graphicx}
\usepackage{adjustbox}
\graphicspath{ {./images/} }

\usepackage[ruled,linesnumbered]{algorithm2e}
\usepackage{multirow}
\usepackage{float}

\usepackage{amsmath,amssymb}
\usepackage{xspace}
\usepackage{centernot}
\usetikzlibrary{positioning,arrows,shapes,backgrounds,calc}

\newcolumntype{L}[1]{>{\raggedright\let\newline\\\arraybackslash\hspace{0pt}}m{#1}}
\newcolumntype{C}[1]{>{\centering\let\newline\\\arraybackslash\hspace{0pt}}m{#1}}
\newcolumntype{R}[1]{>{\raggedleft\let\newline\\\arraybackslash\hspace{0pt}}m{#1}}

\newtheorem{problem2}[thm]{Problem}

\newcommand{\creach}[1]{\mathrel{\xrightarrow#1{}\negthickspace^*}}

\newcommand{\var}[1]{\ensuremath{\mathtt{#1}}}
\newcommand{\vars}{\ensuremath{\mathit{Var}}}

\newcommand{\rscc}{\ensuremath{R_\mathit{scc}}}
\newcommand{\crscc}{\ensuremath{\mathfrak R_\mathit{scc}}}

\newcommand{\cG}{\ensuremath{\mathfrak G}}

\newcommand{\con}{\ensuremath{\mathit{Con}}}
\newcommand{\noncon}{\ensuremath{\mathit{Non}}}
\newcommand{\ccon}{\ensuremath{\mathcal{C}on}}

\newcommand{\Flock}{\ensuremath{F_\mathit{lock}}}
\newcommand{\Block}{\ensuremath{B_\mathit{lock}}}

\newcommand{\cF}{\ensuremath{\mathcal F}}
\newcommand{\cB}{\ensuremath{\mathcal B}}
\newcommand{\cFf}{\ensuremath{\mathcal F_\mathit{open}}}
\newcommand{\cBf}{\ensuremath{\mathcal B_\mathit{open}}}
\newcommand{\cFu}{\ensuremath{\mathcal F_\mathit{paused}}}
\newcommand{\cBu}{\ensuremath{\mathcal B_\mathit{paused}}}
\newcommand{\cW}{\ensuremath{\mathcal W}}

\newcommand{\cblue}{{\color{blue}\mathit{blue}}}
\newcommand{\cred}{{\color{red}\mathit{red}}}

\begin{document}

\title{BDD-Based Algorithm for SCC Decomposition\texorpdfstring{\\}{ }of Edge-Coloured Graphs}

\author[N.~Bene\v{s}]{Nikola Bene\v{s}}
\address{Masaryk University, Brno, Czech Republic}
\email{\{\texttt{xbenes3},\texttt{brim},\texttt{xpastva},\texttt{safranek}\}\texttt{@fi.muni.cz}}

\author[L.~Brim]{Lubo\v{s} Brim}

\author[S.~Pastva]{Samuel Pastva}

\author[D.~Šafránek]{David \v{S}afr{\'a}nek}

\begin{abstract}
Edge-coloured directed graphs provide an essential structure for modelling and analysis of complex systems arising in many scientific disciplines (e.g. feature-oriented systems, gene regulatory networks, etc.). One of the fundamental problems for edge-coloured graphs is the detection of strongly connected components, or SCCs. 

The size of edge-coloured graphs appearing in practice can be enormous both in the number of vertices and colours. The large number of vertices prevents us from analysing such graphs using explicit SCC detection algorithms, such as Tarjan's, which motivates the use of a symbolic approach. However, the large number of colours also renders existing symbolic SCC detection algorithms impractical.

This paper proposes a novel algorithm that symbolically computes all the monochromatic strongly connected components of an edge-coloured graph. In the worst case, the algorithm performs $O(p\cdot n\cdot\log n)$ symbolic steps, where $p$ is the number of colours and $n$ is the number of vertices. 

We evaluate the algorithm using an experimental implementation based on binary decision diagrams (BDDs). Specifically, we use our implementation to explore the SCCs of a large collection of coloured graphs (up to $2^{48}$) obtained from Boolean networks -- a~modelling framework commonly appearing in systems biology.
\end{abstract}

\maketitle

\section*{Introduction}\label{S:one}

In many scientific disciplines, the processing of massive data
sets represents one of the most important computational tasks.
A variety of these data sets can be modelled in terms of very large multi-graphs, augmented by a specific collection of application-dependent edge attributes. These attributes are often abstractly referred to as colours, and the resulting formalism is called an \emph{edge-coloured graph}~\cite{BANGJENSEN1997,bookgraphs79}. Geographic information systems, telecommunications traffic, or internet networks are prime examples of data that are best represented as such edge-coloured graphs. 

For instance, in social networks, coloured edges can be used to link together groups of nodes related by some specific criteria (Sports, Health, Technology, Religion, etc.). In software engineering, one often speaks about feature-oriented systems~\cite{classen2010model}. In this case, colours represent possible combinations of features, altering the system's behaviour.

Our interest in processing huge edge-coloured graphs is primarily motivated by applications taken from systems biology~\cite{tcbb,GiacobbeGGHPP17} and genetics~\cite{DORNINGER94} where we have to deal not only with giant graphs as measured by the number of vertices and edges but also with large sets of colours. In this case, the graph colours represent valuations of numerous parameters that influence the dynamics of a~biological system~\cite{tcbb,BattPCGMJ10,Bernot05}.

Fundamental graph algorithms such as breadth-first search, spanning tree construction, shortest paths, decomposition into strongly connected components (SCCs), etc., are building blocks of many practical applications.  For the edge-coloured graphs, the primary research focus so far has been on some of the ``classical'' coloured graph problems, like the determination of the chromatic index, finding sub-graphs with a specified colour property (the coloured version of the k-linked problem), alternating edge-coloured cycles and paths, rainbow cliques, monochromatic cliques and cycles, etc.~\cite{Das,Akbari,Alon,BANGJENSEN1997,Thomason,Kano}.

To the best of our knowledge, we are not aware of any work on SCC decomposition specifically for edge-coloured graphs, even though this problem has many important applications. For example, in biological systems, strongly connected components represent the so called attractors of the system. In this case, a specific focus is given to terminal (or bottom) SCCs, but non-terminal (transient) SCCs can also be detrimental to the system's long-term behaviour~\cite{long-lived-transients}. Overall, SCCs play an essential role in determining the system's biological properties, since they may correspond, for example, to the specific phenotypes expressed by a~cell~\cite{choo2016efficient}. 

The valuation of parameters (e.g. the presence of certain genes or external stimuli) in such systems is then represented as edge colours in the state-transition graph. The knowledge of SCCs and how their structure depends on parameters is vital for understanding various biological phenomena~\cite{deritei2016principles,Li16}. 
Other applications where investigation of attractors is crucial include predictions of the global climate change~\cite{Steffen18} or predictions of spreading of infectious diseases such as COVID-19~\cite{Matouk20}.

There is a serious computational problem related to the processing of massive edge-coloured graphs (or even the non-coloured ones) that significantly affects the tractability of SCC decomposition. The graphs often cannot be handled using standard (explicit) representations, since they are too large to be kept in the main memory. Various approaches have been considered to deal with such giant graphs: distributed-memory computation, symbolic data structures for graph representation, or storing the graphs in external memory. We review these approaches in more detail in the related work section.

In~\cite{cmsb2017,ICFEM19} we have initially attacked the SCC decomposition problem for massive edge-coloured graphs by developing a parallel, semi-symbolic algorithm for detection of bottom SCCs. The algorithm uses symbolic structures to represent sets of parameters, while the graph itself is represented explicitly. However, the results have shown that the parallel semi-symbolic algorithm is often not sufficient to tackle graphs representing real-world problems practically. These findings have motivated us to propose a new, entirely symbolic approach.

In this paper,  we consider \emph{edge-coloured multi-digraphs}, i.e., multi-digraphs such that each directed edge has a colour and no two parallel (i.e., joining the same pair of vertices) edges have the same colour. Here, we refer to such graphs simply as \emph{coloured graphs}. For coloured graphs, we can define several notions of strongly connected components involving colours. We consider the simplest case, where the  SCCs are \emph{monochromatic}, that is all their edges have the same colour~\cite{Kiraly14}. This choice is motivated by the application in systems biology, as mentioned above.

\paragraph{Contribution} 
We  propose a novel fully symbolic algorithm for detecting \emph{all}
monochromatic strongly connected components in a coloured graph. This algorithm is in practice significantly faster than what is achievable by na\"ively executing a symbolic SCC decomposition
algorithm for each colour separately. This is because in many applications,
the edges are largely shared among individual
colours~\cite{tcbb} and our algorithm is capable of exploiting this fact.
The algorithm conceptually follows the \emph{lock-step} reachability
approach by Bloem et al.~\cite{bloem2000} for purely monochromatic digraphs. The key
new ingredients behind our algorithm are a~careful orchestration of
the forward and backward reachability for different colours, and
a~colour-aware selection of the pivot set.

\subsection*{Structure of the paper} In Section~1, we recall the notions of strongly connected components and  edge-coloured digraphs, and we state the coloured SCC decomposition problem. In Section~2,  we first briefly introduce the forward-backward decomposition algorithm and the lock-step algorithm for monochromatic graphs. After that, we present the coloured SCC decomposition algorithm together with the proof of correctness and complexity analysis. 

In Section~\ref{section:boolean-networks}, we introduce Boolean networks, discuss their symbolic encoding, and show how they can be translated into coloured graphs suitable for SCC-decomposition. Subsequently, Section~\ref{section:implementation} discusses several practical improvements to the main algorithm (saturation, trimming, and parallelism) which help it scale to larger models, and thus be more practically viable. Finally, Section~4 evaluates the main algorithm (including the improved variants) using a collection of large, real-world Boolean networks. A conclusion is provided in the last section.

This article is an extended version of an article that appeared in the Proceedings of TACAS 2021~\cite{tacas21}.  We  extend the information provided in the TACAS Proceedings with more in-depth technical details of the algorithm and related proof, including explanation of its key steps. Moreover, we extend the implementation and evaluation sections to give the reader more information on how the performance of the algorithm can be improved, and how the algorithm performs using a variety of real-world case studies.

\subsection*{Related Work}

The detection of SCCs in (monochromatic) digraphs is a well-known problem
computable in linear time.
Best serial (explicit) algorithms are
Kosaraju-Sharir~\cite{SHARIR1981} and Tarjan~\cite{Tarjan}, which are both  inherently based on depth-first search. However, these algorithms do not scale for large graphs, e.g., those encountered in model-checking, when using explicit graph representation.
Therefore, alternative approaches to such SCC decomposition have been proposed (e.g. I/O efficient, parallel, or symbolic algorithms).

The algorithm of Jiang~\cite{JIANG1993} gives an I/O-efficient alternative
based on a~combination of depth-first and breadth-first search.

Efficient parallel, distributed-memory algorithms avoid the inherently sequential DFS step~\cite{REIF}
in several different ways. The Forward-Backward algorithm~\cite{FWBW} employs
a~divide-and-conquer approach relying on picking a pivot state and
splitting the graph in three independent (SCC-closed) parts. The
approach of Orzan~\cite{Orzan} uses a different distribution scheme
called a colouring transformation, employing a set of prioritised colours to
split the graph into many parts at once. The OWCTY-Backward-Forward (OBF) approach is
proposed in~\cite{Barnat09}. It recursively splits the graph in a number of independent sub-graphs called OBF slices and applies to each slice the One-Way-Catch-Them-Young (OWCTY) technique. In~\cite{Slota14}, the authors utilise
variants of the Forward-Backward and Orzan's algorithms for optimal
execution on shared-memory multi-core platforms. Finally, Bloemen et
al.~\cite{Bloemen16} present an on-the-fly parallel algorithm utilising a swarm of DFS searches, showing promising speed-up for large graphs containing large SCCs. On another end, GPU-accelerated approaches to computing
SCCs have been addressed for example in~\cite{BarnatBBC11,HRO13,LI2014,Dragan14}.

Computing SCCs of (monochromatic) digraphs symbolically is another way to handle giant graphs and has been thoroughly explored
in literature. As in the case of efficient parallelisation, depth-first
search is not feasible in the symbolic framework~\cite{GentiliniPP08}. In~consequence, many DFS-based algorithms cannot be
easily revised to work with symbolically represented graphs. An algorithm based on forward and backward reachability performing $\mathcal{O}(n^2)$ symbolic steps was presented by Xie and Beerel in~\cite{xie2000}. Bloem et al.~present an improved $\mathcal{O}(n \cdot \log n)$ algorithm in~\cite{bloem2000}. Finally, an $\mathcal{O}(n)$ algorithm was presented by Gentilini et al.~in~\cite{gentilini2003,GentiliniPP08}. This bound has been proven to be tight in~\cite{chatterjee2018}. In~\cite{chatterjee2018}, the authors argue that the algorithm from~\cite{gentilini2003} is optimal even when considering more fine-grained complexity criteria, like the diameter of the graph and the diameters of the individual components. Ciardo et al.~\cite{Ciardo11} use the idea of saturation~\cite{Ciardo06} to speed up state exploration within the Xie-Beerel algorithm, and show a saturation-based technique for computing the transitive closure of the graph's edge relation. 

Besides these generic algorithms, there have also been symbolic SCC decomposition methods to deal with large graphs generated specifically by Boolean networks~\cite{MizeraIEEE,YUAN2019}. However, these primarily target detection of bottom SCCs. Methods in this area are also often incomplete, for example focusing on detection of single-state or small bottom SCCs~\cite{zhang-small-attractors}. As such, they generally perform better than an exhaustive symbolic SCC detection in their respective application domains, but are inherently limited in scope.

\section{Problem Definition}\label{S:two}

As we have already stated in the introductory section, the SCC decomposition problem for edge-coloured graphs has remained mostly unexplored until now. We thus start this paper by introducing and formalising the notion of \emph{coloured SCC decomposition} itself and state some of its basic properties.

Before giving exact definitions, it might be instructive
to discuss the substance of the coloured SCC decomposition intuitively. There are
several ways of capturing the notion of a~``coloured connected component''.
One of them is that of a colour-connectivity first introduced by Saad~\cite{Saad92}. It is based on
alternating paths in which successive edges differ in colour. However, there is no
unique, universally acceptable notion of a coloured component.

In the biological applications we have in mind (i.e.~Boolean networks), we want to identify a coloured component as a~coloured collection of SCCs---a~collection where for every colour there is
a set of all relevant monochromatic SCCs. Such a setting leads us to represent SCCs
in the form of a relation. To that end, we first introduce such a relation for
monochromatic graphs (Section~\ref{sec:monographs}) and afterwards extend it
to edge-coloured graphs (Section~\ref{sec:colourgraphs}). The relation-based
approach gives us also the advantage of allowing a feasible symbolic encoding
of the problem.

\subsection{Graphs and Strongly Connected Components}
\label{sec:monographs}

Let us first recall the standard definitions of a~directed graph and its
strongly connected components:

\begin{defi}\label{def:graph}
	A \emph{directed graph} is a tuple $G = (V, E)$ where $V$ is a set of graph \emph{vertices} and $E \subseteq V \times V$ is a set of graph \emph{edges}.
\end{defi}

We are going to use the word \emph{graph} to mean \emph{directed graph} in
the following.
We write $u \to v$ when $(u, v) \in E$ and $u \to^* v$ when $(u, v) \in E^*$,
the reflexive and transitive closure of $E$. We say that $v$ is
\emph{reachable} from $u$ if $u \to^* v$.
The reachability relation allows us to decompose a~graph into strongly
connected components, defined as follows:

\begin{defi}
	In a~graph $G = (V, E)$, a \emph{strongly connected component} (SCC)
	is a~maximal set $W \subseteq V$ such that for all $u, v \in W$, $u
	\to^* v$ and $v \to^* u$. For a fixed $v \in V$, we write $SCC(G, v)$
	to denote the SCC of $G$ that contains $v$.
\end{defi}

If the graph $G$ is clear from the context, we can simply write $SCC(v)$.
A set of vertices $S \subseteq V$ is said to be \emph{SCC-closed} if every SCC
$W$ is either fully contained inside $S$ ($W \subseteq S$), or in its
complement ($W \subseteq V \setminus S$). Notice that given a vertex $v$, the
set of all vertices reachable from $v$, as well as the set of all vertices
that can reach $v$, are both SCC-closed.

A pivotal problem in computer science is to find the SCC decomposition of~$G$.
As mentioned above, we represent the decomposition in the form of an
\emph{equivalence relation} $\rscc$ such that the individual SCCs are exactly
the equivalence classes of $\rscc$. The relation-based formulation of the SCC
decomposition problem is the following:

\begin{problem2}[SCC decomposition]
	Given a graph $G = (V, E)$, find the~\emph{SCC decomposition relation}
	$\rscc \subseteq V \times V$ such that $(u,v) \in \rscc$ if and only if $SCC(u) = SCC(v)$.
\end{problem2}

Note that $SCC(u)$ can be obtained by fixing the first attribute of $\rscc$, i.e.
$SCC(u) = \{ v \mid (u, v) \in \rscc \}$.
We refer to such operation as \emph{section} and denote it in the following way: $SCC(u) = \rscc(u, \_)$ (the concept is properly formalised later as part of Fig.~\ref{tab:operations}).
Here, $u$ is the specific value of an attribute at which the section is taken,
and $\_$ is used in place of the attributes that remain unchanged. Such
notation naturally extends to arbitrary relations.

\subsection{Coloured SCC Decomposition Problem}\label{sec:colourgraphs}

We now lift the formal framework to the coloured setting. An edge-coloured graph can be
seen as a succinct representation of several different graphs, all sharing the same
set of vertices.
To emphasise the difference from the standard graphs (i.e. Definition~\ref{def:graph}), we sometimes call the standard graphs
\emph{monochromatic}.

\begin{defi}
	An \emph{edge-coloured directed multi-graph} (coloured graph for short) is a
	tuple $\cG = (V, C, E)$ where $V$ is a set of vertices, $C$ is a set of
	colours and $E \subseteq V \times C \times V$ is a~coloured edge relation.
\end{defi}

We also write $u \xrightarrow{c} v$ whenever $(u, c, v) \in E$ and
use $\creach{c}$ to denote the reflexive and transitive closure
of $\xrightarrow{c}$. We say that $v$ is $c$-reachable from $u$ if
$u \creach{c} v$, i.e.~there is a~path from $u$ to $v$ using
only $c$-coloured edges.
By fixing a~colour $c \in C$ and keeping only the $c$-coloured edges (with the
colour attribute removed), we obtain a~monochromatic graph
$\cG(c) = (V, \{(u, v) \mid (u, c, v) \in E\})$. We call this graph the
\emph{monochromatisation of $\cG$ with respect to $c$}.
Intuitively, one can view the elements of $C$ as a type of graph
parametrisation where the edge structure of the graph changes based on the
specific $c \in C$.

The SCC decomposition relation $\rscc$ is extended to the coloured SCC
decomposition relation $\crscc$ by relating every colour $c\in C$ with all
SCCs of the monochromatisation $\cG(c)$. In consequence, the SCC decomposition
problem is then lifted to the coloured SCC decomposition problem as follows:

\begin{problem2}[Coloured SCC decomposition]
	Given a coloured graph $\cG = (V, C, E)$, find the \emph{coloured SCC
		decomposition relation} $\crscc \subseteq V \times C \times V$ satisfying
	$(u,c,v) \in \crscc$ if and only if $(u,v) \in \rscc$ of $\cG(c)$.
\end{problem2}

From this definition, we can immediately observe the following properties
about the relationship of $\crscc$ with the terms which we have defined before:
\begin{itemize}
	\item $\rscc$ of a monochromatisation $\cG(c)$ is exactly the section
	$\crscc(\_, c, \_)$;
	\item $SCC(\cG(c), v)$ is exactly the section $\crscc(v, c, \_)$, or equivalently, $\crscc(\_, c, v)$ (since $\crscc$ and $\rscc$ are symmetric with regards to $V$).
\end{itemize}
From this, it should be immediately apparent that $\crscc$ contains all
components of the underlying monochromatisations.

\section{Algorithm}

Conceptually, our algorithm follows the \emph{lock-step} reachability approach by Bloem~\cite{bloem2000} for monochromatic graphs. The lock-step algorithm itself is based on the basic forward-backward decomposition algorithm~\cite{xie2000}.
In this section, we first briefly introduce these two algorithms to explain better the key ideas behind our approach and, in particular, to explain the main difficulties encountered in employing the concepts of these algorithms to edge-coloured graphs.
Although the algorithms were originally presented as producing a~set of SCCs,
we reformulate them slightly using the equivalent relation-based approach as
explained in the previous section.
After that, we present the coloured SCC decomposition algorithm.
However, before we dive into the algorithmics, let us first briefly discuss
the computation model we are using.

\subsection{Symbolic Computation Model}\label{ssec:symb}

As a complexity measure of our algorithm, we consider the number of
symbolic steps, or more specifically, symbolic set and relation operations that the
algorithm performs. As is customary, we assume that sets of vertices ($V$) and colours ($C$) can be represented symbolically (for example, using reduced ordered binary decision diagrams~\cite{bryant86}) as well as any relations over these sets. In particular, we often talk about \emph{coloured vertex sets}, by which we mean the subsets of $V \times C$.

\begin{figure}[t]
	\centering
	\setlength\tabcolsep{8 pt}
	\renewcommand{\arraystretch}{1.4}
	\begin{tabular}{ | C{3cm} | C{4cm} | C{6cm} | }
		\hline
		\multicolumn{3}{|c|}{Standard set operations} \\ \hline
		pick element & $\textsc{Pick}(A)$ & arbitrary $x \in A$ \\
		union & $A \cup B$ & $\{ x \mid x \in A \lor x \in B \}$ \\
		intersection & $A \cap B$ & $\{ x \mid x \in A \land x \in B \}$ \\
		difference & $A \setminus B$ & $\{ x \mid x \in A \land x \not\in B \}$ \\
		product & $A \times B$ & $\{ (x,y) \mid x \in A \land y \in B \}$ \\
		\hline
		\multicolumn{3}{|c|}{Relation manipulation ($R \subseteq S_1 \times \ldots \times S_n$)} \\ \hline
		$i$-th section at $x$ & $\sigma_i(x, R)$ & $\{ (y_1, \ldots, y_{i-1}, y_{i+1}, \ldots, y_n) \mid (y_1, \ldots, y_{i-1}, x, y_{i+1}, \ldots, y_n) \in R \}$ \\
		existential quantification of the~$i$-th element & $\exists_i(R)$ & $\bigcup_{x \in S_i} \sigma_i(x, R)$ \\
		swap & $\textsc{Swap}(R \subseteq A \times B)$ & $\{ (y, x) \in B \times A \mid (x, y) \in R \}$ \\
		\hline
		\multicolumn{3}{|c|}{Derived operations ($G = (V, E), \cG = (V, C, E)$)} \\ \hline
		colours & $\textsc{Colours}(A \subseteq V \times C)$ &
		$\exists_1(A)$\\
		pre-image & $\textsc{Pre}(G, A \subseteq V)$& $\exists_2((V \times A) \cap E)$ \\
		post-image & $\textsc{Post}(G, A \subseteq V)$ & $\exists_1((A \times V) \cap E)$ \\
		coloured pre-image & $\textsc{Pre}(\cG, A \subseteq V \times C)$ & $\exists_3((V \times \textsc{Swap}(A)) \cap E)$ \\
		coloured post-image & $\textsc{Post}(\cG, A \subseteq V \times C)$ & $\textsc{Swap}(\exists_1((A \times V) \cap E))$  \\
		coloured join & $\textsc{Join}(A \subseteq V \times C)$ & $(V \times \textsc{Swap}(A)) \cap (A \times V)$ \\
		\hline
	\end{tabular}
	\caption{Summary of symbolic operations that appear in the presented algorithms. The derived operations can be implemented using the standard and relational operations. However, typically they also have a slightly more efficient direct implementations.}%
	\label{tab:operations}
\end{figure}

Aside from normal set operations (union, intersection, difference, product and element selection), we also require some basic relational operations, all of which we outline in Figure~\ref{tab:operations}. These extra operations tend to appear in other applications as well (such as symbolic model checking~\cite{BurchCMDH92}), and are thus typically already available in mature symbolic computation packages.

Finally, there are several derived operators that are partially
specific to our application to coloured graphs. However, these can be
constructed using standard set and relation operations.
The intuitive meaning of the derived operators is as follows:
$\textsc{Colours}$ returns all the colours that appear in the given coloured
vertex set. $\textsc{Pre}$ and $\textsc{Post}$ compute the pre- and post-image
of a (monochromatic or coloured) set of vertices, i.e.~the set of successors
or predecessors of all the vertices in the given set, respectively.
Finally, $\textsc{Join}$ takes a~coloured vertex set $A$ and computes
the set $\{(u, c, v) \mid (u, c) \in A, (v, c) \in A\}$.

\subsection{Forward-Backward Algorithm}

To symbolically compute the SCCs of a~graph $G = (V, E)$, Xie and
Beerel~\cite{xie2000} observed that for any vertex $v \in V$,
the intersection $W = F \cap B$
of the forward reachable vertices $F = \{ v' \in V \mid v \to^* v' \}$ and the
backward reachable vertices $B = \{ v' \in V \mid v' \to^* v \}$ is exactly
the strongly connected component of $G$ which contains $v$.

The algorithm thus picks an arbitrary \emph{pivot} $v \in V$, and divides the
vertices of the graph into four disjoint sets: $W$, $F \setminus W$, $B
\setminus W$ and $V \setminus (F \cup B)$. This is illustrated graphically in Figure~\ref{fig:algorithms}~(left). The set $W$ is then immediately
reported as an SCC of the graph, and added into the component relation: $\rscc
\gets \rscc \cup (W\times W)$. It is easy to see that every other SCC is
fully contained within one of the three remaining sets (they are SCC-closed), and the algorithm thus
recursively repeats this process independently in each set.

The correctness of the algorithm follows from the initial observation and the
fact that every vertex eventually appears in $W$ (either as a~pivot or as
a~result of $F \cap B$). In the worst case, the algorithm performs $O(|V|^2)$
symbolic steps, since every vertex is picked as a~pivot at most once and the
computation of $F$ and $B$ requires at most $O(|V|)$
\textsc{Pre}/\textsc{Post} operations.

\begin{figure}
	\centering
	\includegraphics[width=0.4\linewidth]{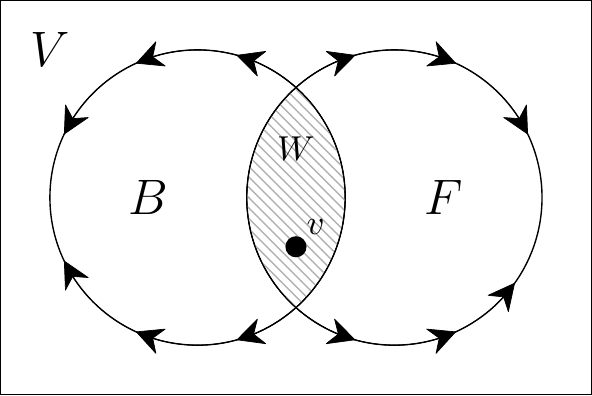}
	\hspace{0.2cm}
	\includegraphics[width=0.4\linewidth]{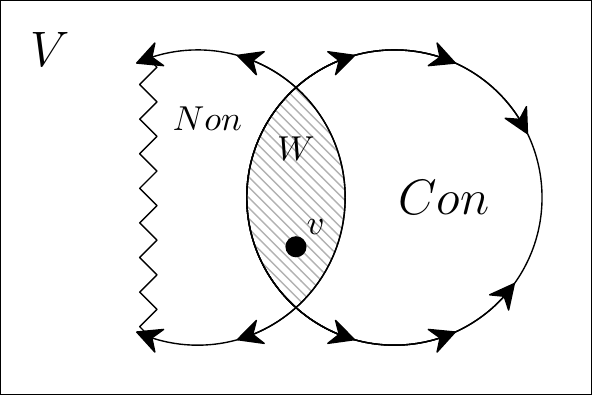}
	\caption{Illustration of the difference between the forward-backward algorithm (left) and the lock-step algorithm (right). On the left, we fully compute both backward ($B$) and forward ($F$) reachable sets from the pivot $v$, identifying $W$ as $F \cap B$. On the right, without loss of generality, assume $F$ is fully computed first. It thus becomes converged ($Con$) and the computation of $B$ ($Non$) is stopped before it is fully explored. }
	\label{fig:algorithms}
\end{figure}

\subsection{Lock-step Algorithm}

To improve the efficiency of the forward-backward algorithm, the lock-step
approach~\cite{bloem2000} uses another important observation:
To compute $W$, it is not
necessary to fully compute both $F$ and $B$; only the smaller (in terms of diameter) of the two sets
needs to be entirely known. With this observation, the computation of $F$ and
$B$ can be modified in the following way:
Instead of computing $F$ and $B$ one after the other, the computation is
\emph{interleaved} in a step-by-step manner (dovetailing). When one of the
sets is fully computed, the computation of the second set is stopped.
Let us call the computed set \emph{converged} and denote it by $\con$, and the
unfinished set \emph{non-converged} and denote it by $\noncon$. This situation is illustrated in Figure~\ref{fig:algorithms}~(right).

However, even when $\con$ is fully known, we still need to finish the
computation of states in $\noncon$ that are inside $\con$ to discover the
whole component $W$. This is necessary if there are vertices $w$ in $W$ whose
forward distance from $v$ (i.e.~the length of the path $v \to^* w$) is short
while their backward distance (the length of the path $w \to^* v$) is long,
or vice versa. Such vertices are thus only discovered in one of the two
reachability procedures and still need to be discovered by the other one
to identify the whole component.
However,
an important observation is that only the vertices already inside $\con$ need
to be considered in this phase.

After this, the SCC can be identified and reported just as in the
forward-backward algorithm. Finally, the recursion now continues in sets $\con
\setminus W$ and $V \setminus \con$. This is due to $\noncon$ being not fully
computed; we cannot guarantee that no SCC overlaps outside of $\noncon$ ($\noncon$ is not necessarily SCC-closed).

The algorithm is still correct because every vertex is eventually either
picked as a pivot or discovered in some $W$. Furthermore, due to the way
$\con$ and $\noncon$ are computed guarantees that $W$ is still a whole SCC\@.
In terms of complexity, the algorithm performs $O(|V| \cdot \log |V|)$
symbolic steps in the worst case. To see why this is true, we may observe that
every vertex appears in $W$ exactly once, and that the smaller of the two
sets $\con \setminus W$ and $V \setminus \con$, let us call it $S$,
is always smaller than $\frac{|V|}{2}$. The authors then argue that
the price of every iteration can be attributed (up to a multiplicative
constant) to the vertices in $S \cup W$ and that every vertex appears in $S$
at most $O(\log |V|)$-times.

\subsection{Coloured Lock-step Algorithm}

When developing an algorithm for coloured graphs, one needs to deal with multiple
challenges which do not appear for monochromatic graphs and require
careful consideration. In the following, we refer to the pseudocode in
Algorithm~\ref{algo:symbolic}.

\begin{algorithm}
	\SetKwProg{Fn}{Function}{}{}
	\Fn{\upshape\textsc{ColouredSCC}$(\cG = (V, C, E))$}{
		$\crscc \subseteq (V \times C \times V) \gets \emptyset$\;
		$\textsc{Decomposition}(\cG, \crscc, V \times C)$\;
		\Return $\crscc$\;
	}
	\BlankLine
	\Fn{\upshape\textsc{Decomposition}$(\cG = (V, C, E), \crscc \subseteq (V \times C \times V), \mathcal{V} \subseteq (V \times C))$}{
		\lIf{$\mathcal{V} = \emptyset$}{\Return}
		$\cF, \cB, \cFf, \cBf \subseteq (V \times C) \gets \textsc{Pivots}(\mathcal{V})$\;
		$\cFu, \cBu \subseteq (V \times C) \gets \emptyset$\;
		$\Flock, \Block \subseteq C \gets \emptyset$\;
		\While{$\Flock \cup \Block \subset \textsc{Colours}(\mathcal{V})$}{\label{alg:ls-start}
			$\cFf \gets (\textsc{Post}(\cG, \cFf) \cap \mathcal V) \setminus \cF$\;\label{alg:lockstep-post}
			$\cBf \gets (\textsc{Pre}(\cG, \cBf) \cap \mathcal V) \setminus \cB$\;\label{alg:lockstep-pre}
			$\Flock \gets \Flock \cup (\textsc{Colours}(\mathcal V) \setminus \textsc{Colours}(\cFf) \setminus \Block)$\;
			$\Block \gets \Block \cup (\textsc{Colours}(\mathcal V) \setminus \textsc{Colours}(\cBf) \setminus \Flock)$\;\label{alg:lock}
			$\cFu \gets \cFu \cup (\cFf \cap (V \times \Block))$\;
			$\cBu \gets \cBu \cup (\cBf \cap (V \times \Flock))$\;
			$\cFf \gets \cFf \setminus (V \times \Block)$\;
			$\cBf \gets \cBf \setminus (V \times \Flock)$\;
			$\cF \gets \cF \cup \cFf$\;
			$\cB \gets \cB \cup \cBf$\;
		}\label{alg:ls-end}
		$\ccon \subseteq V \times C \gets (\cF \cap (V \times \Flock)) \cup (\cB \cap (V \times \Block))$\;\label{alg:con}
		$\cFf \gets \cFu \cap \ccon$\;\label{alg:con-apply}
		$\cBf \gets \cBu \cap \ccon$\;\label{alg:con-apply-2}
		\While{$\cFf \ne \emptyset \lor \cBf \ne \emptyset$}{\label{alg:rest-start}
			$\cFf \gets (\textsc{Post}(\cG, \cFf) \cap \ccon) \setminus \cF$\;
			$\cBf \gets (\textsc{Pre}(\cG, \cBf) \cap \ccon) \setminus \cB$\;
			$\cF \gets \cF \cup \cFf$\;
			$\cB \gets \cB \cup \cBf$\;
		}\label{alg:rest-end}
		$\cW \subseteq V \times C \gets \cF \cap \cB$\;\label{alg:scc}
		$\crscc \gets \crscc \cup \textsc{Join}(\cW)$\;
		$\textsc{Decomposition}(\cG, \crscc, \mathcal{V} \setminus \ccon)$\;
		$\textsc{Decomposition}(\cG, \crscc, \ccon \setminus \cW)$\;
	}
	\BlankLine
	\Fn{\upshape\textsc{Pivots}$(\mathcal{V})$}{
		$\mathcal{P} \subseteq (V \times C) \gets \emptyset$; $\mathcal{V}' \subseteq (V \times C) \gets \mathcal{V}$\;
		\While{$\mathcal{V}' \ne \emptyset$}{
			$(v, c) \gets \textsc{Pick}(\mathcal{V}')$\;
			$\mathcal{P} \gets \mathcal{P} \cup ( \{ v \} \times \sigma_1(v, \mathcal{V'})) $\;
			$\mathcal{V'} \gets \mathcal{V'} \setminus (V \times \textsc{Colours}(\mathcal{P}))$\;
		}
		\Return $\mathcal P$\;
	}
	\caption{Symbolic Coloured SCC Decomposition}\label{algo:symbolic}
\end{algorithm}

An important observation is that the structure of components in the graph can change arbitrarily with
respect to the graph colours. In consequence, our algorithm cannot simply operate with sets of
graph vertices as the normal algorithm would. To that end, we use the
notion of coloured vertex sets as introduced in
Section~\ref{ssec:symb} where the symbolic operations we perform on
these sets have been described.

\paragraph{Pivot selection} Initially, the algorithm starts with all vertices and colours, i.e.~the full
set $V \times C$. However, as the components are discovered, the intermediate results $\mathcal{V}$ may contain different vertices appearing only for certain subsets of $C$. As a result, we often cannot pick a~single pivot vertex that
would be valid for all considered colours. Instead, we aim to pick a~\emph{pivot set} $P \subseteq V \times C$ such that for every colour
that still appears in $\mathcal V$,
the set contains \emph{exactly} one vertex. Alternatively, one can also
view the pivot set as a (partial) function from $C$ to $V$. This is done
in the \textsc{Pivots} function. In the following discussion of the algorithm,
we write $c$-coloured pivot to mean the vertex $u$ such that
$(u, c)$ is found in the coloured set returned by \textsc{Pivots}
in the current iteration (for all colours still present in $\mathcal{V}$). 

Please note that the presented \textsc{Pivots} routine is rather naive, as it has to explicitly iterate all the pivot vertices, whose number can be substantial in the worst case. However, as presented, it should be easy to implement for basically any type of coloured graphs, regardless of the underlying representation. In the implementation section, we show \textsc{Pivots} can be re-implemented in the domain of BDDs such that it is guaranteed to always require only $\mathcal{O}(\log|V|)$ symbolic operations. 

\paragraph{Coloured lock-step (phase one)} The lock-step reachability procedure also cannot operate as in a standard
graph. First of all, there can be colours where the forward reachability
converges first, as well as colours where this happens for backward
reachability. The algorithm thus has to account for both options
simultaneously. Second, for each colour, the reachability can converge in a
different number of steps. To deal with this problem, we introduce the
$\Flock$ and $\Block$ variables. These store the mutually disjoint sets of
colours for which the forward and backward reachability procedures have already
converged. The lock-step procedure then terminates when $\Flock$ and $\Block$
contain all the colours that appear in $\mathcal V$.

Throughout the algorithm, we keep track of several coloured-set variables.
The first two, $\cF$ and $\cB$, represent the forward and backward reachable
sets, respectively.
This means that for every colour $c$ present in $\mathcal V$, if $u$ is the
$c$-coloured pivot, every $(v, c) \in \cF$ satisfies $u \creach{c} v$
and every $(v, c) \in \cB$ satisfies $v \creach{c} u$.
Furthermore, if $c \in \Flock$ then $\cF$ contains exactly all such pairs;
similarly for $\Block$ and $\cB$.

We say that a coloured vertex pair $(v, c)$ has been \emph{forward expanded}
or \emph{backward expanded} in the current iteration of the algorithm, if
there has been a call to the \textsc{Post} or \textsc{Pre} symbolic operation
with $(v, c)$ being an element of the coloured set argument.
To track which reachable coloured vertices are to be expanded later,
also called the \emph{frontiers} of the reachability sets,
we have the four variables $\cFf$, $\cFu$, $\cBf$, $\cBu$.

The frontier of $\cF$ is the union $\cFf \cup \cFu$. The sets $\cFf$ and $\cFu$
are disjoint: $\cFf$ involves those colours for which the
lock-step reachability procedure has not finished yet, i.e.~the colours that
are neither in $\Flock$ nor in $\Block$,
while $\cFu$ represents
the part of the frontier whose exploration is currently paused due to the fact
that its colours are in $\Block$.
Note that there may be no pair $(v, c)$ of the forward frontier with
$c \in \Flock$ as that means that the exploration of the $c$-coloured
forward-reachable set is complete.
A symmetric role is played by the sets $\cBf$ and $\cBu$.

In the first while loop (lines~\ref{alg:ls-start}--\ref{alg:ls-end}), we
compute the reachability sets in the lock-step manner. Once a~reachability set
is completed for some colours (i.e.,~there are no vertices to expand with those
colours), we add the colours to the corresponding $\Flock$ or $\Block$
variable. Note that we ensure that $\Flock$ and $\Block$ remain disjoint even
if the two reachability procedures converged at the same time for certain
colours---see line~\ref{alg:lock}.
We use $\Flock$ and $\Block$ to split the newly computed frontier sets into
the parts that are to be expanded in the next iteration ($\cFf$, $\cBf$)
and the parts currently left unexpanded ($\cFu$, $\cBu$).

Note that during the computation of $\textsc{Post}$ and $\textsc{Pre}$ on lines~\ref{alg:lockstep-post} and~\ref{alg:lockstep-pre}, we intersect the resulting set with $\mathcal{V}$. This step is not necessary for correctness, but as the algorithm divides $V \times C$ into smaller sets in each recursive call to \textsc{Decomposition}, it can happen that the set of states \emph{reachable} from $\mathcal{V}$ is substantially larger than $\mathcal{V}$ itself. In such cases, this intersection effectively restricts the computation of $\textsc{Post}$ and $\textsc{Pre}$ to the sub-graph of $\cG$ induced by $\mathcal{V}$.

\paragraph{Component identification (phase two)} After the first while loop terminates, we compute the set $\ccon$ that is an analogue
for the converged set of the original lock-step algorithm
(line~\ref{alg:con}).
As already suggested above and unlike the original
algorithm, this set cannot be just $\cF$ or $\cB$, but is instead a~mixture of
both, depending on the converged colours. To~compute this set, we use the
$\Flock$ and $\Block$ variables. 

Once $\ccon$ is computed, $\cFf$ and $\cBf$ are restarted using the converged portion of $\cFu$ and $\cBu$ (lines~\ref{alg:con-apply} and~\ref{alg:con-apply-2}). The second while loop (lines~\ref{alg:rest-start}--\ref{alg:rest-end}) can then
complete the unfinished forward and backward reachability set, now restricted to
the inside of the converged set. The intersection of $\cF$ and $\cB$ then
forms a~coloured set $\cW$ with the property that for all
$c \in \textsc{Colours}(\mathcal V)$, $\cW(\_, c)$ is a~strongly
connected component of $\cG(c)$. We create the corresponding relation
using the \textsc{Join} operation, add this relation to the resulting
$\crscc$, and recursively call the whole procedure with
$\mathcal V \setminus \ccon$ and $\ccon \setminus \cW$ as the
base sets.

\paragraph{Comments on the coloured approach} Let us note that there is possibly another approach to processing coloured graphs. Instead of trying
to work with all colours still appearing in the coloured vertex set at once,
we could fork a~new recursive procedure whenever the colour set
splits due to the differences in the graph structure. For example, instead
of picking multiple coloured vertices as pivots, one could pick a single
vertex with a valid subset of colours and then address the remaining colours
in a separate recursive call. Similarly, instead of a single recursive $\textsc{Decomposition}$ call with $\ccon \setminus \cW$, we could consider two calls, one with the $\mathcal{F}$ portion of $\ccon$ and the other with the $\mathcal{B}$ portion of $\ccon$ (note that these are colour-disjoint since each colour can converge only in one of the two sets).

While such an approach could be to some extent
beneficial in a massively parallel environment where each recursive call can
be executed independently on a~new CPU, the amount of forking in large systems
will soon become unreasonable. More importantly, it defeats the purpose of
the symbolic representation, which aims to minimise the number of symbolic
operations. 

\begin{figure}
	\centering
	\begin{tikzpicture}
	\begin{scope}[every node/.style={minimum width=1.2em,inner sep=0,draw,thick,circle,font={\small\itshape}},x=1.5cm,y=1.5cm]

	\foreach \p / \x / \y in {/0/0,f2/-2/-4,b2/+2/-4,f3/-2/-6,b3/+2/-6,f4/-2/-8,b4/+2/-8,g/0/-10}
	\node (\p a) at (\x - 1, \y) {a};

	\node[draw=red] (f1a) at (-3, -2) {a};
	\node[draw=red] (b1a) at (+1, -2) {a};

	\foreach \p / \x / \y in {f1/-2/-2,b1/+2/-2,f2/-2/-4,b2/+2/-4,f3/-2/-6,b3/+2/-6,f4/-2/-8,b4/+2/-8,g/0/-10}
	\node (\p b) at (\x, \y) {b};

	\node[double] (b) at (0, 0) {b};

	\foreach \p / \x / \y in {/0/0,f2/-2/-4,f3/-2/-6,b3/+2/-6,f4/-2/-8,b4/+2/-8,g/0/-10}
	\node (\p c) at (\x + 1, \y) {c};

	\node[draw=blue] (f1c) at (-1, -2) {c};
	\node[draw=blue] (b1c) at (+3, -2) {c};
	\node[draw=red] (b2c) at (+3, -4) {c};

	\foreach \p / \x / \y in {/0/0,f1/-2/-2,b2/+2/-4,f3/-2/-6,b3/+2/-6,f4/-2/-8,b4/+2/-8,g/0/-10}
	\node (\p d) at (\x - 1, \y - 1) {d};

	\node[draw=red] (b1d) at (+1, -3) {d};
	\node[draw=red] (f2d) at (-3, -5) {d};

	\foreach \p / \x / \y in {/0/0,f2/-2/-4,b2/+2/-4,b3/+2/-6,f4/-2/-8,b4/+2/-8,g/0/-10}
	\node (\p e) at (\x, \y - 1) {e};

	\node[draw=blue] (f1e) at (-2, -3) {e};
	\node[draw=red] (b1e) at (+2, -3) {e};
	\node[draw=red, dashed] (f3e) at (-2, -7) {e};

	\foreach \p / \x / \y in {/0/0,f2/-2/-4,f3/-2/-6,f4/-2/-8,b4/+2/-8,g/0/-10}
	\node (\p f) at (\x + 1, \y - 1) {f};

	\node[draw=blue] (f1f) at (-1, -3) {f};
	\node[draw=red] (b1f) at (+3, -3) {f};
	\node[draw=blue, dashed] (b2f) at (+3, -5) {f};
	\node[draw=blue, dashed] (b3f) at (+3, -7) {f};
	\end{scope}

	\begin{scope}[on background layer]
	\foreach \n in {b,f1b,b1b,f2b,b2b,b2c,b2f,f3b,f3e,b3b,b3c,b3f,f4b,f4e,f4f,b4b,b4c,b4e,b4f,gb,ge,gf} {
		\begin{scope}
		\clip ($(\n) + (-.5, -.5)$)
		-- ($(\n) + (.5, .5)$)
		-- ($(\n) + (.5, -.5)$)
		--cycle;
		\fill[red!50!white] (\n) circle (0.6em);
		\end{scope}
		\begin{scope}
		\clip ($(\n) + (-.5, -.5)$)
		-- ($(\n) + (.5, .5)$)
		-- ($(\n) + (-.5, .5)$)
		--cycle;
		\fill[blue!50!white] (\n) circle (0.6em);
		\end{scope}
	}
	\foreach \n in {f1a,b1a,b1d,b1e,b1f,f2a,f2d,b2a,b2d,b2e,f3a,f3d,b3a,b3d,b3e,f4a,f4d,b4a,b4d,ga,gd} {
		\fill[red!50!white] (\n) circle (0.6em);
	}
	\foreach \n in {f1c,f1e,f1f,b1c,f2c,f2e,f2f,f3c,f3f,f4c,gc} {
		\fill[blue!50!white] (\n) circle (0.6em);
	}
	\end{scope}

	\foreach \p in {,f1,b1,f2,b2,f3,b3,f4,b4} {
		\draw[->,>=stealth',blue,thick,shorten >=1pt,shorten <=1pt]
		(\p a) edge[bend right=15] (\p d)
		(\p b) edge[bend right=15] (\p c)
		edge[bend left=15] (\p e)
		edge[bend left=10] (\p f)
		(\p c) edge[bend right=15] (\p b)
		(\p d) edge[bend right=15] (\p e)
		(\p e) edge[bend right=15] (\p f)
		(\p f) edge[bend left=15] (\p c)
		;

		\draw[->,>=stealth',red,thick,shorten >=1pt,shorten <=1pt]
		(\p a) edge[bend left=15] (\p b)
		edge[bend left=15] (\p d)
		(\p b) edge[bend left=15] (\p a)
		(\p c) edge[bend left=15] (\p f)
		(\p d) edge (\p b)
		edge[bend left=15] (\p e)
		(\p e) edge[bend left=15] (\p b)
		edge[bend left=15] (\p f)
		(\p f) edge[bend left=10] (\p b)
		;
	}

	\draw[->,>=stealth',blue,thick,shorten >=1pt,shorten <=1pt]
	(gb) edge[bend right=15] (gc)
	edge[bend left=15] (ge)
	edge[bend left=10] (gf)
	(gc) edge[bend right=15] (gb)
	(ge) edge[bend right=15] (gf)
	(gf) edge[bend left=15] (gc)
	;

	\draw[->,>=stealth',red,thick,shorten >=1pt,shorten <=1pt]
	(ga) edge[bend left=15] (gb)
	edge[bend left=15] (gd)
	(gb) edge[bend left=15] (ga)
	(gd) edge (gb)
	edge[bend left=15] (ge)
	(ge) edge[bend left=15] (gb)
	edge[bend left=15] (gf)
	(gf) edge[bend left=10] (gb)
	;

	\node at (-4, -2) {$\mathcal F$};
	\node at (+4, -2) {$\mathcal B$};

	\node[align=center] at (3.5, -0.5) {\small$\textsc{Pivots}(\mathcal V) = $\\
		\small$\{(b, \cblue), (b, \cred)\}$};

	\draw[dashed,gray] (-5, -8) -- (5, -8);
	\draw[dashed,gray] (-5, -11) -- (5, -11);
	\node[anchor=south] at (0,-8) {\small$\cblue \in \Flock$};
	\node[anchor=south] at (0,-11) {\small$\cred \in \Block$};
	\node[align=center] at (0, -12.7) {\small after\\\small second\\\small phase};

	\end{tikzpicture}
	\caption{An illustration of the algorithm execution.
		There are two colours: $\cred$ and $\cblue$.
		The left column represents the forward reachability part; the right column
		represents the backward reachability part.
		Filled nodes are contained in $\cF$, $\cB$ respectively.
		Nodes with solid coloured border are in $\cFf$, $\cBf$; nodes with dashed
		coloured border are in $\cFu$, $\cBu$.
		The bottom-most picture represents the resulting coloured set $\mathcal W$.
	}\label{fig:example}
\end{figure}

\paragraph{Example} The execution of one iteration of the algorithm is illustrated in
Figure~\ref{fig:example}. Here, we have an edge-coloured graph with six vertices
and two colours (red and blue).
The top-most picture represents the initial situation after we have chosen
the pivots; in this case, $\{(b, \mathit{blue}), (b, \mathit{red})\}$.
The next four rows illustrate the first phase (the first while loop) of the
algorithm. After the second iteration of the loop, the blue colour becomes
locked in $\Flock$, and thus $(f, \mathit{blue})$ is not expanded in the
backward reachability procedure. This is illustrated by its dashed outline.
After the third iteration of the loop,
the red colour becomes locked in $\Block$ and thus the first phase ends.
In the second phase, both reachability procedures continue from the paused
coloured vertices (dashed outlines); the result is seen in the fifth row.
The intersection of the two reachable sets (i.e.~the coloured set $\mathcal W$)
is then illustrated in the bottom-most picture.
The algorithm would now continue with the coloured sets
$\mathcal{V} \setminus \ccon = \{ (a, \mathit{blue}), (d, \mathit{blue}) \}$
and $\ccon \setminus \cW = \{ (c, \mathit{red}) \}$.

\subsection{Correctness and Complexity of the Coloured Lock-step Algorithm}

\begin{thm}\label{thm:correctness}
	Let $\cG = (V, C, E)$ be a~coloured graph. The coloured lock-step
	algorithm terminates and computes the coloured SCC decomposition relation~$\crscc$.
\end{thm}
\begin{proof}
We first show that the set $\cW$ computed in line~\ref{alg:scc}
indeed contains one SCC for every colour $c \in \textsc{Colours}(\mathcal V)$
and that the recursive calls of \textsc{Decomposition} preserve the property
that $\mathcal V$ is SCC-closed with respect to all colours.

Let us assume that $\mathcal V$ is SCC-closed and let us take an arbitrary
$c \in \textsc{Colours}(\mathcal V)$. The function \textsc{Pivots} chooses
a~set that contains exactly one pair whose colour is $c$, let us call this
pair $(v, c)$. Let us further assume that $c$ is assigned into $\Flock$
first (the case with $\Block$ is completely symmetric).

Let us now choose an arbitrary vertex $w$ such that $v$ and $w$ are in the
same SCC of $\cG(c)$, i.e.~$v \to^* w$ and $w \to^* v$.
As the first while loop
finishes, $\cF$ contains all the pairs of the form $(u, c) \in \mathcal V$
where $u$ is reachable from $v$ in $\cG(c)$. Thus, it also contains $(w, c)$
due to the fact that $\mathcal V$ is SCC-closed.
Now, either $(w, c) \in \cB$, or there exists a vertex $x$ such that
$w \to^* x$, $x \to^* v$ in $\cG(c)$ and $x \in \cBu$.
This means that $(w, c)$ is added to $\cB$ in the second while loop.
In both cases, both $(v, c)$ and $(w, c)$ are then added to
$\cW$. As the vertex choices were arbitrary, this proves that
the SCC of $v$ in $\cG(c)$ is contained in $\cW$.
Furthermore, if $(y, c) \in \cW$ for an arbitrary $y$, then
$v \to^* y$ and $y \to^* v$ in $\cG(c)$, which means that $y$ is in
$SCC(\cG(c), v)$. This proves that $\cW$ contains exactly one
SCC for every colour $c \in \textsc{Colours}(\mathcal V)$.

We now argue that $\ccon$ is SCC-closed with respect to all colours.
This immediately implies that both $\mathcal V \setminus \ccon$ and
$\ccon \setminus \cW$ are SCC-closed.
Let us assume that there is a~colour $c \in \textsc{Colours}(\mathcal V)$
and two vertices $v$, $w$ in the same SCC of $\cG(c)$ such that
$(v, c) \in \ccon$, but $(w, c) \not\in \ccon$.
Let us assume that $c \in \Flock$ (as above, the case of $\Block$ is
completely symmetrical).
Then $(v, c) \in \cF$ after the first while loop finishes. This also means
that $(w, c) \in \cF$ as the forward reachability procedure is completed for
$c$ and thus $(w, c) \in \ccon$, a~contradiction.

What remains is to show that the algorithm terminates and that every
SCC is eventually found. Termination is trivially proved by the fact
that size of the set $\mathcal V$ always decreases in recursive calls:
both $\cW$ and $\ccon$ are non-empty because they contain the
initial pivot set as a~subset.
Clearly, a representant of every SCC of every monochromatisation $\cG(c)$ is
eventually chosen as a~pivot. Together with the above reasoning, this implies
that the algorithm is correct. 
\end{proof}

\begin{thm}\label{thm:complexity}
	Let $|V|$ be the number of vertices in the coloured graph and let
	$|C|$ be the number of colours. The coloured lock-step algorithm
	performs at most $\mathcal O(|C|\cdot |V|\cdot \log |V|)$ symbolic steps.
\end{thm}
\begin{proof}
	Let us first note that all the derived operations defined in
	Figure~\ref{tab:operations} use only a~constant number of the basic symbolic
	operations. As we are considering asymptotic complexity here, we can view all
	the operations in Figure~\ref{tab:operations} as elementary symbolic steps.

	We first make the observation that each vertex may be chosen as a~part of the
	pivot set at most $|C|$ times. Clearly, once a~vertex is included in the pivot
	set with a~set of colours $C'$, then, $\{v\} \times C'$ is a subset of first $\ccon$, and later $\cW$
	(due to the monotonicity of the construction of $\mathcal F$ and $\mathcal B$). Therefore, the elements of $\{v\} \times C'$ do not appear in subsequent recursive
	calls. Since a single vertex-colour pair cannot be returned by \textsc{Pivots} twice, it means that the total cumulative complexity of all the calls to the \textsc{Pivots} routine is bounded by $O(|C|\cdot|V|)$. We can therefore exclude them from the rest of the complexity analysis.

	We now consider the complexity of a~single call to \textsc{Decomposition}
	without the subsequent recursive calls. Let us now select one of the
	colours for which the lock-step reachability procedure
	(lines~\ref{alg:ls-start}--\ref{alg:ls-end}) finished last, i.e.,~one of the
	colours that have been added to $\Flock$ or $\Block$ in the final
	iteration of the loop. Let us call this colour $c$.
	Recall that $\sigma_2(c, \mathcal X)$ is the set of vertices with colour
	$c$ in a~coloured set $\mathcal X$.

	Let us denote by $W$ the monochromatic SCC discovered for $c$, i.e. $W := \sigma_2(c, \cW)$, and let
	$S$ be the smaller of $\sigma_2(c, \mathcal V \setminus \ccon)$
	and $\sigma_2(c, \ccon \setminus \cW)$.
	Clearly $S$ contains at most $|V|/2$ vertices.
	Let $k = |S \cup W|$.
	We now argue that the number of symbolic steps in a~given call
	(without the recursive calls) is bounded by $\mathcal O(k)$. This is because in a lock-step algorithm, the call to \textsc{Decomposition} must explore the discovered SCC itself (i.e. $W$), and the smaller of the forward or backward reachable sets from this SCC (i.e. $S$) -- intuitively, its complexity should be thus bounded by the size of these two sets.

	Assume w.l.o.g.~that $c \in \Flock$ (a completely symmetric argument solves
	the case $c \in \Block$).
	Then after the first while loop finishes, we have $\sigma_2(c, \ccon) = \sigma_2(c, \mathcal F)$.
	If $S$ is $\sigma_2(c, \ccon \setminus \cW)$
	then $k$ is the size of $\sigma_2(c, \mathcal F)$ (and thus also $\sigma_2(c, \ccon)$), since $\sigma_2(c, \mathcal{F})$ consists of $\sigma_2(c, \ccon \setminus \cW)$ (the set $S$) and $\sigma_2(c, \cW)$ (the discovered SCC).
	Each iteration of the first while loop puts at least one vertex with
	colour $c$ into $\mathcal F$; otherwise $c$ would not have finished in the last iteration. This means that the loop runs for at most
	$k$ iterations. This also means that the size of $\sigma_2(x, \mathcal X)$ for
	all colours $x$ and $\mathcal X \in \{ \mathcal F, \mathcal B \}$ is also
	bounded by $k$ after the first while loop finishes, which in turn means that the second while loop cannot make more
	than $O(k)$ steps.

	If $S$ is $\sigma_2(c, \mathcal V \setminus \ccon)$ instead,
	let us define $B := \sigma_2(c, \mathcal B)$ right after the first
	while loop has finished. We know that $B \subseteq S \cup W$:
	if a vertex $v$ was in $B \setminus S$, then it would have to be in $\ccon$ (i.e. $(v, c) \in \ccon$). Due to our initial assumption of $c \in \Flock$ (w.l.o.g), we then also have $(v,c) \in \mathcal
	F$ which dictates $v \in W$. Consequently, we see that any vertex $v \in B$ must be either in $S$ or in $W$, arriving at $B \subseteq S \cup W$.
	Again, each iteration of the first while loop puts at least one vertex
	with colour $c$ into $\mathcal B$; otherwise $c$ would have been
	in $\Block$ before it appeared in $\Flock$. Similarly to the previous
	case, this means that both while loops run for at most $O(k)$ steps.

	The rest of the argument uses amortised reasoning, in a~way similar to the
	proof in~\cite{bloem2000}. Note that each vertex is going to be an element of
	the set $W$ as described above at most $|C|$ times (once for each colour).
	Furthermore, each vertex is going to be an element of the set $S$ as described
	above at most $|C|\cdot\log|V|$ times: for each colour, the vertex can be an
	element of the smaller of the two sets at most $\log|V|$ times. As the cost of
	each single call can be charged to the vertices in $S \cup W$ as explained
	above, it is sufficient to charge each vertex the total cost of
	$|C| +|C|\cdot\log|V|$. Together, this means that the total number of symbolic
	steps is bounded by $O(|C|\cdot|V|\cdot\log|V|)$.
\end{proof}

Note that the upper bound established by Theorem~\ref{thm:complexity} is
no better than the one we would get if we split the coloured graph into
its monochromatic constituents and processed each
separately using the original lock-step algorithm~\cite{bloem2000}.
We remark, however, that the practical complexity of the coloured approach can be
much smaller. Indeed, the complexity analysis in the previous proof focused
on a~single colour, omitting the fact that SCCs for many other colours are
found at the same time. In cases where the edges are largely shared among
the colours, which is true in many applications, the coloured algorithm has the
potential to significantly outperform the parameter-scan approach. The
situation is similar to that of the coloured model checking; see the
observations made in~\cite{tcbb}.

\section{Symbolic Computation with Boolean Networks}
\label{section:boolean-networks}

The algorithm as presented in the previous section is completely agnostic to the properties of the underlying system, as long one provides an implementation of all the necessary symbolic operations. However, to empirically test its performance, we need to pick such an implementation, which typically entails analysis of a specific class of systems. 

In this paper, we consider Boolean networks~\cite{KAUFFMAN,Bernot05,SCHWAB,THOMAS}, specifically asynchronous Boolean networks, which represent a popular discrete modelling framework in systems biology~\cite{Brim2013, Grieb_2015}. Due to incomplete biological knowledge, the dynamics of a Boolean network can by often only partially known. This uncertainty can be then captured using coloured directed graphs. In this section, we introduce Boolean networks and show how they can be translated into coloured graphs suitable for SCC-decomposition.

Asynchronous Boolean networks are especially challenging for symbolic analysis. It is a well-known fact, that using symbolic structures (e.g BDDs) to explore very large state spaces gives good results for synchronous systems, but shows its limits when trying to tackle asynchronicity (see e.g.~\cite{DBLP:conf/forte/CouvreurT05}).

\subsection{Boolean networks with inputs} A Boolean network (BN), as the name suggests, consists of $n$ Boolean \emph{variables} $s_1, \ldots, s_n$ which together describe the state of the network. The dynamics of the network can also depend on additional $m$ Boolean \emph{inputs} $c_1, \ldots, c_m$ (sometimes also called \emph{constants}, or \emph{logical parameters}), whose value is assumed to be fixed, but generally unknown. The valuations of these inputs correspond to the colours of our Kripke structure.

Each network variable $s_i$ is equipped with a Boolean update function $b_i: \{0,1\}^n \times \{0,1\}^m \to \{0,1\}$ that updates the variable based on the state of the network, and the values of its inputs. We assume that the variables are updated \emph{asynchronously}, meaning that during every state transition, exactly one variable is updated.

Such a network with inputs defines a coloured graph where $V = \{0,1\}^n$, $C = \{0,1\}^m$, and for every $c \in C$, we have that $u \xrightarrow{c} v$ if and only if $u \not= v$ and $v = u[u_i \mapsto b_i(u, c)]$ for some $i \in [1,n]$. That is, $v$ is equal to $u$ where the $i$-th variable is updated with the output of function $b_i$. Because all variables and inputs are Boolean, this structure has a~fairly straightforward symbolic representation in terms of binary decision diagrams, as we later demonstrate.

Note that in practice, we often work within a~subset of biologically relevant colours, denoted as~$Valid$ (i.e. not every possible valuation of $c_1, \ldots, c_m$ may be biologically admissible). In the algorithms, this is implicitly reflected such that the set of all possible colours $C$ corresponds to the set $Valid$ instead of the set of \emph{all} possible valuation (i.e. $\{0,1\}^m$) if demanded by the application at hand.

\subsection{Partially specified Boolean networks} A Boolean network with inputs allows us to easily encode a wide range of biochemical systems in a~machine friendly format. However, for systems with a high degree of uncertainty, it often fails to capture this uncertainty in a~way understandable to a human reader.

To mitigate this issue, we consider \emph{partially specified} Boolean networks that allow us to explicitly mark parts of the update functions as unknown. Specifically, let us assume that $f_1^{(a_1)}$, $f_2^{(a_2)}$, $\ldots$ are symbols standing in for some uninterpreted (fixed but arbitrary) Boolean functions (here, $a_i$ denotes their arity). A partially specified Boolean network then consists of $n$ Boolean variables and $p$ uninterpreted Boolean functions. In such a~network, every update function $b'_i$ is specified as a Boolean expression that can use the function symbols $f_1, \ldots, f_p$.

This type of formalism is often easier to comprehend, as the uncertainty in dynamics is tied to the update functions instead of inputs (if desired, input can be still expressed using uninterpreted functions of arity zero). It is not immediately clear how such a~network should be represented symbolically though.

One option is to translate a partially specified network into a~BN with inputs. Any uninterpreted function $f_i^{(a)}$ can be encoded in terms of $2^a$ Boolean inputs $c_1^{i}, \ldots, c_{2^a}^i$ if we consider that input $r_j^i$ denotes the output of $f_i^{(a)}$ in the $j$-th row of its truth table. Formally, this translation can be achieved using a~repeated application of the following expansion rule:
\begin{align*}
	f(\alpha_1, \ldots, \alpha_a) \equiv (\alpha_1 \Rightarrow f'_1(\alpha_2, \ldots, \alpha_a)) \land (\neg\alpha_1 \Rightarrow f'_2(\alpha_2, \ldots, \alpha_a))
\end{align*}
Here, $f'_1$ and $f'_2$ are fresh uninterpreted functions of arity $a-1$, and $\alpha_i$ are arbitrary Boolean expressions. Using this rule, we can always convert a partially specified network to a Boolean network with inputs. The number of inputs will be exponential with respect to the arity of the employed uninterpreted functions though (since each application of the rule replaces one uninterpreted function with two, and the depth of the recursive expansion is the arity $a$).

For example, consider the following partially specified Boolean network:
\begin{align*}
	b'_1 &:= x_1 \land f_1^{(1)}(x_2)\\
	b'_2 &:= \neg x_1 \lor f_2^{(2)}(x_1, x_3)\\
	b'_3 &:= (f_3^{(0)} \Leftrightarrow x_3) \land f_2^{(2)}(\neg x_1, x_2)
\end{align*}

It uses three uninterpreted Boolean functions $f_1^{(1)}$, $f_2^{(2)}$, and $f_3^{(0)}$. After performing the aforementioned expansion, and simplifying the resulting expressions slightly for readability, we obtain the following network with logical inputs:

\begin{align*}
	b_1(x, c) =&~x_1 \land (x_2 \Rightarrow c^{1}_{[1]}) \land (\neg x_2 \Rightarrow c^{1}_{[0]})\\
	b_2(x, c) =&~\neg x_2 \lor (((x_1 \land x_3) \Rightarrow c^{2}_{[1,1]}) \land ((x_1 \land \neg x_3) \Rightarrow c^{2}_{[1,0]})\\&\hspace{15.5pt}\land~((\neg x_1 \land x_3) \Rightarrow c^{2}_{[0,1]}) \land ((\neg x_1 \land \neg x_3) \Rightarrow c^{2}_{[0,0]})) \\
	b_3(x, c) =&~(c^{3} \Leftrightarrow x_3) \land ((\neg x_1 \land x_2) \Rightarrow c^{2}_{[1,1]}) \land ((\neg x_1 \land \neg x_2) \Rightarrow c^{2}_{[1,0]})\\&\hspace{57pt}\land((x_1 \land x_2) \Rightarrow c^{2}_{[0,1]}) \land ((x_1 \land \neg x_2) \Rightarrow c^{2}_{[0,0]})
\end{align*}

Here, each $c^i_j$ corresponds to one truth table row of $f_i$, such that $j$ describes the input vector corresponding to said row (i.e. $c^{1}_{[0,1]}$ represents the value of $f_1(0, 1)$).

\subsection{Symbolic Representation of BNs}

As a symbolic representation, a natural choice are Reduced Ordered Binary Decision Diagrams (ROBDD, or simply BDD)~\cite{bryant86}, which can concisely encode Boolean functions or relations of Boolean vectors. Specifically, out implementation leverages the internal tools and libraries provided by the tool AEON~\cite{aeon}.

Since a Boolean network consists of $n$ Boolean variables and $m$ Boolean inputs, any subset of $V$, $C$, or a relation $X \subseteq V \times C$ (a coloured set of vertices) can be seen as a Boolean formula over the network variables and inputs. That is, each network variable and logical input corresponds to one decision variable of the BDD. Here, a pair $(s,c)$ belongs to such a~relation iff it represents a satisfying assignment of this formula $X$. For relations of higher arity, fresh decision variables are created for each component of the relation. Standard set operations as described in Fig.~\ref{tab:operations} then correspond to logical operations on such formulae ($\land \equiv \cap$, $\lor \equiv \cup$, etc.). 

Relation operations are similarly implementable using BDD primitives. In particular, existential quantification of a single decision variable (e.g. $\exists s_i . X$ or $\exists c_j . X$) is a native operation on BDDs. Consequently, existential quantification on relations (as well as \textsc{Colours}) is simply a quantification over all decision variables encoding the specific relation component (i.e. all network variables for $V$, or all logical inputs for $C$). Finally, \textsc{Swap} only influences the way in which a BDD is interpreted -- the actual structure of the BDD is unaffected. 

To encode the network dynamics, notice that every update function $b_i$ can be directly represented as a separate BDD. From such BDDs, we can build one large BDD describing the whole coloured transition relation, which is traditionally used for the computation of \textsc{Pre} and \textsc{Post}. But the symbolic representation of such relation is often prohibitively complex for asynchronous systems. Instead, we compute \textsc{Pre} and \textsc{Post} using partial results for individual variables, which uses more symbolic operations but is less likely to cause a~blow-up in the size of the BDD:
\begin{align*}
	\textsc{VarPost}(\cG, i, \mathcal{X}) & = (\mathcal{X} \land (b_i \centernot\Leftrightarrow s_i))[s_i \mapsto \neg s_i]\\
	\textsc{VarPre}(\cG, i, \mathcal{X}) & = \mathcal{X}[s_i \mapsto \neg s_i] \land (b_i \centernot\Leftrightarrow s_i)\\
	\textsc{Post}(\cG, \mathcal{X}) & = \bigvee_{i \in [1,n]} \textsc{VarPost}(i, \mathcal{X})\\
	\textsc{Pre}(\cG, \mathcal{X}) & = \bigvee_{i \in [1,n]} \textsc{VarPre}(i, \mathcal{X})	
\end{align*}
Here, $[s_i \mapsto \neg s_i]$ is the standard substitution operation, which we use to flip the value of variable $s_i$ in the resulting formula if it does not agree with the output of $b_i$. Note that this operation can be also implemented structurally directly on the BDD by exchanging the children of decision nodes conditioning on $s_i$. Also note that sub-formulae that do not depend on $X$ can be pre-computed once for the whole run of the algorithm, and the version of $\textsc{Pre}$ and $\textsc{Post}$ for monochromatic graphs can be implemented in exactly the same way.

\section{Implementation}
\label{section:implementation}

Finally, let us discuss a number of technical improvements which our algorithm employs in practice, and whose impact we consider in the evaluation section.

\subsection{Pivot Selection}

In Algorithm~\ref{algo:symbolic}, we gave a naive implementation of the $\textsc{Pivots}(\mathcal{X})$ function. Here, we show how to implement it for BDDs in a much more concise way. Note that our approach uses the notation we established earlier for Boolean networks, but is generally applicable to any set or relation of bit-vectors represented using BDDs. 

First, notice that for a single network variable, we can define a similar operation, which we call $\textsc{Pick}(i, \mathcal{X})$:
\begin{equation*}
	\textsc{Pick}(i, \mathcal{X}) = \mathcal{X} \setminus (\mathcal{X} \land \neg s_i)[s_i \gets \neg s_i]
\end{equation*}

Here, we first restrict $\mathcal{X}$ only to the valuations which have $s_i = \mathit{false}$, and then invert the value of $s_i$ (resulting in $s_i$ being always $\mathit{true}$ in the set). Once we subtract these valuations from $\mathcal{X}$, the resulting set then contains a valuation with $s_i = \mathit{true}$ only if it does \emph{not} contain the same valuation with $s_i = \mathit{false}$. Intuitively, for any valuation of the remaining BDD decision variables (i.e. $s_j$ and $c_j$ in our case) that is in $\mathcal{X}$, we just picked a single unique value of $s_i$ (while preferring the value $s_i = \mathit{false}$). 

However, observe that we cannot simply apply $\textsc{Pick}$ to every network variable alone to obtain the result of \textsc{Pivots}. Intuitively, the problem lies in the fact that \textsc{Pick} selects a witness for each variable in isolation, while $\textsc{Pivots}$ considers all network variables as interconnected. We resolve this problem using a different equation, one which eliminates the picked variable in the recursive invocation:
\begin{align*}
	\textsc{Pivots}(\mathcal{X}) &= \textsc{F}(\mathcal{X}, s_1, \ldots, s_n)\\
	\textsc{F}(\mathcal{X}, s_1) &= \textsc{Pick}(1, \mathcal{X})\\
	\textsc{F}(\mathcal{X}, s_1, \ldots, s_k) &= \textsc{Pick}(k, \mathcal{X}) \cap \textsc{F}(\exists s_k. \mathcal{X}, s_1, \ldots, s_{k-1})
\end{align*}

In this equation, the final case $\textsc{F}(\mathcal{X}, s_1)$ is clearly correct, since it simply defers to $\textsc{Pick}(i, \mathcal{X})$. However, to understand why the recursive case $\textsc{F}(\mathcal{X}, s_1, \ldots, s_k)$ is correct, observe the following: Assume that the set $\mathcal{Y} = \textsc{F}(\exists s_k. \mathcal{X}, s_1, \ldots, s_{k-1})$ is computed correctly. That is, for any valuation of the remaining variables, $\mathcal{Y}$ contains a single unique \emph{incomplete witness} valuation of variables $s_1, \ldots, s_{k-1}$. Now, since the BDD representing $\exists s_k . \mathcal{X}$ does not depend on $s_k$, each such unique \emph{witness} must be included in $\mathcal{Y}$ twice: once with $s_k = \mathit{true}$ and once with $s_k = \mathit{false}$. In other words, a single \emph{witness} valuation of $s_1, \ldots, s_{k-1}$ must be tied to two different valuations of the remaining variables, and these valuations are differentiated only by the variable $s_k$.

Now, one of these two valuations is necessarily included in the set $\textsc{Pick}(k, \mathcal{X})$. The other is either missing from $\mathcal{X}$ altogether, or is eliminated by $\textsc{Pick}(k, \mathcal{X})$. As such, computing $\textsc{Pick}(k, \mathcal{X}) \cap \mathcal{Y}$ extends the witness from $s_1, \ldots, s_{k-1}$ to $s_1, \ldots, s_{k}$ by eliminating one of the two aforementioned occurrences of the \emph{incomplete witness}.

Observe that, as opposed to the original naive implementation of \textsc{Pivots}, this implementation only requires $\mathcal{O}(n)$ (i.e. $\mathcal{O}(\log |V|)$) symbolic operations in any case.

\subsection{Saturation}

In~\cite{Ciardo06}, and later in greater detail within~\cite{Ciardo11}, Ciardo et al. show that when the system is asynchronous, it may be much easier to compute reachable sets (and consequently SCCs) by applying only one transition (e.g. denoted $t_1$) at a time. Once applying $t_1$ cannot add new states to the reachable set, another transition (e.g denoted $t_2$) can be considered, respecting the order in which the affected variables appear in the symbolic data structure (Ciardo et al. employ multivalued decision diagrams, but the principle also applies to BDDs). If the application of other transitions causes that we can again add new states using $t_1$, the process starts anew and $t_1$ is ``saturated'' again.

In the comparison presented in~\cite{Ciardo11}, only the Xie-Beerel $\mathcal{O}(|V|^2)$ algorithm is used with saturation enabled, while the lock-step algorithm is used as given in~\cite{bloem2000}. However, we argue that saturation can be also beneficial in the lock-step algorithm. 

\paragraph{Asymptotic complexity} Unfortunately, combining lock-step with saturation disrupts the $\mathcal{O}(|V| \cdot \log|V|)$ asymptotic complexity of the algorithm. To see why this is the case, observe that classical symbolic reachability (i.e. a fixed-point algorithm iterating the $\textsc{Post}$ procedure) requires $\mathcal{O}(|V|)$ steps to explore a graph. Meanwhile, a reachability procedure employing saturation needs $\mathcal{O}(|V||T|)$ operations, where $|T|$ is the number of distinct transitions. 

This is caused by the fact that saturation needs to check up to $|T|$ transitions to discover a vertex. For example, consider an asynchronous graph employing transitions $t_1, \ldots, t_n$ such that $t_1$ and $t_n$ are alternated on a path of length $\mathcal{O}(|V|)$. Between considering $t_1$ and $t_n$, saturation will attempt each of the $|T|$ transitions, which are useless on this path, but still consume a symbolic operation.

Consequently, this $|T|$ factor trickles down into the complexity of both the Xie-Beerel and lock-step algorithms if saturation is used, as both ultimately rely on some form of reachability to discover the graph vertices. The complexity of the coloured algorithms is then similarly affected.

\paragraph{Saturation and lock-step} The main idea of how saturation is applied in a coloured lock-step algorithm (for Boolean networks) is shown in Algorithm~\ref{algo:saturation}. The algorithm presents a helper function which performs \emph{one reachability step}, similar to what is performed by the $\textsc{Post}$ function. However, in this algorithm, only one transition is fired for each colour (we assume the iteration follows the order of variables as they appear in the symbolic representation, which benefits saturation). Additionally, a set $R$ of colours that could not perform a step is computed. A similar procedure can be considered for backwards reachability, simply replacing $\textsc{VarPost}$ with $\textsc{VarPre}$. 

Note that there is a slight discrepancy between Algorithm~\ref{algo:saturation} and the intuitive description of saturation that we gave earlier. In particular, we see that during a \textsc{NextStep} operation, a transition for each variable is triggered at most once, as opposed to the original description, where a transitions are fired repeatedly. This is caused by the simple nature of Boolean networks: In a BN, a single transition always modifies a single Boolean variable. Consequently, no new states can be discovered by firing a single transition multiple times in sequence. For other asynchronous systems, $\textsc{VarPost}$ may need to be modify to apply the corresponding transition repeatedly.

Additionally, note that we use the set $R$ to ensure that $\textsc{VarPost}$ (i.e. firing of a single transition) is executed only for colours for which we have not found a successor yet using some of the previously considered transitions. This is necessary to ensure that in each invocation of \textsc{NextStep}, each colour present in $\mathcal{F}$ is either advanced by one step (using exactly one transition), or is reported as converged within the returned set $R$.

Using this process, we can replace the $\textsc{Pre}/\textsc{Post}$ procedures in the main lock-step algorithm (lines 11 and 12 of Algorithm~\ref{algo:symbolic}). The $R$ sets computed here are then used to update $F_{lock}$ and $B_{lock}$ (lines 13 and 14), as they exactly represent the converged colours that do not need further computation. A similar modification is necessary for the second while loop (lines 25-29), but here the sets of remaining colours $R$ are not needed.

\begin{algorithm}
	\SetKwProg{Fn}{Function}{}{}
	\Fn{\textsc{NextStep}$(\mathfrak{G}, \mathcal{F})$}{
		$R \gets \textsc{Colours}(\mathcal{F})$\;
		\For{$\var{A} \in \vars$}{
			$\mathcal{S} \gets \textsc{VarPost}(\mathfrak{G}, \var{A}, (\mathcal{F} \cap V) \times R)$\;
			$R \gets R \setminus \textsc{Colours}(\mathcal{S})$\;
			$\mathcal{F} \gets \mathcal{F} \cup \mathcal{S}$\;
			\If{$R = \emptyset$}{\textbf{break}\;}
		}		
		\Return $\mathcal (\mathcal{F}, R)$\;
	}
	\caption{Main idea of the lock-step-saturation approach. The algorithm extends $\mathcal{F}$ with one additional reachability step, and returns a set of colours locked in this iteration ($R$).}\label{algo:saturation}
\end{algorithm}

\subsection{Trimming and Parallelism}

Most graphs typically contain a large number of trivial SCCs that introduce unnecessary overhead to the main algorithm. To avoid this overhead, we additionally perform a trimming step before each invocation of \textsc{Decomposition}. Trimming consists of repeatedly removing all vertices which have no outgoing or no incoming edges and is employed by most symbolic SCC algorithms on standard directed graphs as well. 

The coloured analogue of trimming is straightforward, as it can be achieved using \textsc{Pre} and \textsc{Post} operations just as in the non-coloured case. For a coloured set of vertices $\mathcal{V}$, operation $\textsc{Post}(\cG, \textsc{Pre}(\cG, \mathcal{V}) \cap \mathcal{V}) \cap \mathcal{V}$ returns only the vertices which have at least one predecessor in $\mathcal{V}$. The successor variant simply exchanges the \textsc{Post} and \textsc{Pre} operations. 

As such, applying this operation to each $\mathcal{V}$ until a fixed-point is reached before \textsc{Decomposition} is invoked eliminates the undesired trivial SCCs. Since the total number of steps performed collectively by all such fixed-point computations is bounded by $|C||V|$ (the total number of removable vertex-colour pairs), this does not impact the overall asymptotic complexity of the algorithm.

In some cases, we have observed that the symbolic representation is able to handle the SCC computation but explodes during trimming. The algorithm then times-out during trimming, even though useful information about SCCs could be obtained if the trimming was skipped or postponed. To avoid this issue, we enforce an extra condition that a trimming procedure is terminated prematurely if the computed BDDs are more than twice the size (in terms of BDD decision nodes) of the initial set. 

Additionally, the lock-step algorithm can be rather trivially parallelised. The recursive \textsc{Decomposition} calls operate on independent coloured vertex sets and can be therefore deferred to separate threads. Since the body of the \textsc{Decomposition} method is rather complex, this can be done easily with a queue guarded by a mutex which is shared between all threads (i.e. the synchronisation overhead is negligible due to the long running time of \textsc{Decomposition}). Finally, a simple termination detection procedure is needed to ensure that idle threads do not terminate prematurely while decomposition is still running.

Note that most BDD packages are not internally thread-safe, as they share decision node memory across different BDD objects. In our experiments, this aspect is handled by cloning the set $\mathcal{V}$ corresponding to each recursive invocation, plus the symbolic representation of the BN necessary to compute $\textsc{Post}$ and $\textsc{Pre}$. As such, the memory used to represent BDDs manipulated by each thread is completely independent from other threads.

\section{Experimental Evaluation}

To test the algorithm, we compiled a benchmark set of Boolean networks from the CellCollective~\cite{helikar2012cell} and GINsim~\cite{chaouiya2012} model databases. Since the models in these databases contain fully specified networks, uninterpreted functions were introduced into existing models by pseudo-randomly erasing parts of the existing update functions. 

While this process is to some extent artificial, we believe it to be a good approximation of the model development process, where at some point, the structure of the network is already established, but its dynamics are still not fully determined. Using this process, we obtained a collection of networks ranging between $2^{20}$ and $2^{50}$ in the size of the coloured graph (i.e. $|V \times C|$). Note that for each graph, we consider only a subset of possible input valuations that is biologically relevant with respect to the established network structure. For example, the first model (i.e.~\cite{sanchez2017modeling}) admits $2^{48}$ input valuations, but only $2^{19}$ are biologically relevant due to constraints on function monotonicity. 

A complete overview of the employed models is given in Table~\ref{tab:models}. For each model, we give the number of discovered non-trivial components as an interval, because each colour can correspond to a different number of components. We employ a 24h timeout for all experiments.

\begin{table}
	\caption{The considered benchmark models. Here, $n$ is the number of BN variables, $m$ is the number of logical inputs (after expansion of uninterpreted functions), $|C|$ is the number of all biologically relevant colours (input valuations), and $|V \times C|$ is the size of the whole biologically relevant coloured state space. Finally, \#SCC gives the number of detected non-trivial SCCs. Note that this number varies depending on input valuation, and is thus given as a range.}
	\label{tab:models}
	\setlength\tabcolsep{8 pt}
	\renewcommand{\arraystretch}{1.5}
	\begin{tabular}{ | c | c | c | c | c | c | }
		\hline
		\textbf{Model name} & $n$ & $m$ & $|C|$ & $|V \times C|$ & \#SCC \\\hline
		{\small Asymmetric Cell Division~\cite{sanchez2017modeling}} & $5$ & $48$ & $\sim2^{19}$ & $\sim2^{24}$ & 1-13 \\ \hline
		
		{\small Reduced TCR Signalisation~\cite{klamt2006methodology}} & $10$ & $46$ & $\sim2^{14}$ & $\sim2^{24}$ & 36-115 \\ \hline
		
		{\small Budding Yeast (Orlando)~\cite{orlando2008global}} & $9$ & $54$ & $\sim2^{16}$ & $\sim2^{27}$ & 1-16 \\ \hline
		
		{\small Budding Yeast (Irons)~\cite{irons2009logical}} & $18$ & $44$ & $\sim2^{17}$ & $\sim2^{35}$ & 2-5568 \\ \hline
		
		{\small Tumor Cell Migration~\cite{cohen2015mathematical}} & $20$ & $44$ & $\sim2^{15}$ & $\sim2^{35}$ & 436-379308 \\ \hline
		
		{\small T-cell Differentiation~\cite{mendoza2006method}} & $23$ & $40$ & $\sim2^{15}$ & $\sim2^{38}$ & 41728-43264 \\ \hline
		
		{\small WG Signalling Pathway~\cite{mbodj2013logical}} & $26$ & $38$ & $\sim2^{22}$ & $\sim2^{48}$ & 0 \\ \hline
		
		{\small Full TCR Signalisation~\cite{klamt2006methodology}} & $30$ & $48$ & $\sim2^{17}$ & $\sim2^{47}$ &  48-1087 \\ \hline
	\end{tabular}
\end{table}

\begin{table}
	\caption{Overview of runtime for different version of the SCC detection algorithm. The times (\texttt{hours:minutes:seconds}) refer to the total runtime of the SCC decomposition procedure for the basic lock-step, lock-step with saturation, and lock-step with saturation and parallelism, with \texttt{DNF} representing a time-out after 24-hours. }\label{tab:results}
	\centering
	\setlength\tabcolsep{8 pt}
	\renewcommand{\arraystretch}{1.5}
	\begin{tabular} { | c | c | c | c | }
		\hline
		\textbf{Model Name} & \textbf{Parallel} & \textbf{Satur.} & \textbf{Lock-step}  \\ \hline
		
		{\small Asymmetric Cell Division~\cite{sanchez2017modeling}} & \texttt{00:05} & \texttt{00:10} & \texttt{00:15} \\ \hline
		
		{\small Reduced TCR Signalisation~\cite{klamt2006methodology}} & \texttt{00:04} & \texttt{00:45} & \texttt{01:12} \\ \hline
		
		{\small Budding Yeast (Orlando)~\cite{orlando2008global}} & \texttt{06:29} & \texttt{06:50} & \texttt{11:21} \\ \hline
		
		{\small Budding Yeast (Irons)~\cite{irons2009logical}} & \texttt{15:14} & \texttt{2:53:16} & \texttt{3:28:44} \\ \hline
		
		{\small Tumor Cell Migration~\cite{cohen2015mathematical}} & \texttt{40:10} & \texttt{18:34:16} & \texttt{DNF} \\ \hline
		
		{\small T-cell Differentiation~\cite{mendoza2006method}} & \texttt{16:10:41} & \texttt{DNF} & \texttt{DNF} \\ \hline
		
		{\small WG Signalling Pathway~\cite{mbodj2013logical}} & \texttt{1:18:38} & \texttt{1:23:37} & \texttt{1:42:12} \\ \hline
		
		{\small Full TCR Signalisation~\cite{klamt2006methodology}} & \texttt{4:49:04} & \texttt{DNF} & \texttt{DNF} \\ \hline
		
	\end{tabular}
\end{table}

The experiments were performed on a 32-core AMD Threadripper workstation with 64GB of RAM memory. All tested models are available in our source code repository.\footnotemark[3] \footnotetext[3]{\url{https://github.com/sybila/biodivine-lib-param-bn/tree/lmcs}} Note that the smaller models ($<2^{30}$) should be easy to process even on a less powerful machine; however, the larger models can require substantial amount of memory.

For each model, we have tested the lock-step algorithm as presented in the main part of this paper (\emph{Lock-step} in Table~\ref{tab:results}), an enhanced version with saturation enabled (\emph{Satur.} in Table~\ref{tab:results}), and a parallel implementation which also includes saturation (\emph{Parallel} in Table~\ref{tab:results}). In all algorithms, we employ the trimming optimisation.

From the results, we can see that parallelisation improves the performance of the algorithm significantly: in case of models with a large number of SCCs, we see an up-to 30x speed-up, comparing \emph{Parallel} and \emph{Satur.} in Table~\ref{tab:results}. On the other hand, when the number of SCCs is small (such as~\cite{orlando2008global}), the speed-up is understandably minimal, since the number of independent recursive calls is also small.

As expected, the total number of SCCs has a significant impact on the performance of the algorithm (e.g.~\cite{irons2009logical} and~\cite{cohen2015mathematical}) overall, since the number of calls to \textsc{Decomposition} increases. Furthermore, we see that our ``coloured saturation'' indeed provides a performance benefit. However, this improvement is mostly incremental.

After further analysis, we discovered that the whole algorithm is often limited by the performance of the trimming procedure, rather than reachability procedures though. In~particular, the use of saturation has significantly reduced the size of symbolic representation during computation of reachability, however the symbolic representation still performs rather poorly (at least for Boolean networks) during trimming. This limits the performance of the whole method, since all the considered graphs contain a large portion of trivial SCCs. Furthermore, in many cases the number of iterations needed to completely trim a set of states is substantial. This leads us to believe there is still space for improvement in terms of SCC detection in large Boolean networks, even without parameters.

Finally, we examined the benefit of processing all colours simultaneously versus a naive parameter scan approach, where each monochromatic case is handled separately. To do so, we considered various pseudo-random monochromatisations of the studied models and processed these using our algorithm. Here, we observe that for the four models with at least $20$ variables, no computation for any of the monochromatic models finished in under one second (with T-cell differentiation typically requiring more than one minute due to the relatively large number of components). 

Consequently, we can extrapolate that computing the full coloured SCC decomposition using such naive parameter scan would require more than 10 ours for each model (and $10+$ days in the case of T-cell differentiation). This approach could be to some extent beneficial in a massively parallel environment (hundreds or thousands of CPUs), but the coloured approach clearly scales better in setups where resources are more limited.

\section{Conclusions}

This paper presents a fully symbolic algorithm for detecting all monochromatic strongly connected components in edge-coloured graphs. The work has been motivated by systems sciences, namely systems biology, where the need for efficient automated analysis of components in large graphs with a large sets of coloured edges is emerging. The algorithm combines several ideas inspired by existing state-of-the-art algorithms for SCC decomposition in a~non-trivial way. We believe this is the first fully symbolic algorithm aiming to solve the problem efficiently.

The experimental evaluation has shown that the algorithm can handle large, real-world systems that would be otherwise too large to fit into the memory of a conventional workstation ($>2^{32}$), and that the performance of the algorithm can be further improved using saturation and parallelisation. Finally, the algorithm has a strong potential to be significantly faster
than iterating a standard algorithm for SCC decomposition executed on all monochromatic sub-graphs one-by-one.

\bibliography{bib-tacas2021}

\bibliographystyle{alphaurl}

\end{document}